\documentclass[pdflatex,sn-mathphys]{sn-jnl}

\jyear{2021}%

\theoremstyle{thmstyleone}%
\newtheorem{theorem}{Theorem}
\theoremstyle{thmstyletwo}%

\theoremstyle{thmstylethree}%
\newtheorem{definition}{Definition}%

\usepackage{savesym}
\savesymbol{MOD}
\savesymbol{C}
\savesymbol{EXP}
\savesymbol{p}
\savesymbol{R}
\savesymbol{A}
\usepackage{complexity}
\restoresymbol{COM}{MOD}
\restoresymbol{COM}{C}
\restoresymbol{COM}{EXP}
\restoresymbol{COM}{p}
\restoresymbol{COM}{R}
\restoresymbol{COM}{A}
\newcommand{\ket}[1]{\vert #1 \rangle}
 
\newcommand{\braket}[2]{\langle #1 \vert #2 \rangle}

\newcommand{\eps}[0]{\varepsilon}
\definecolor{darkgreen}{RGB}{29, 148, 65}

\begin{document}

\title[Quantum-resistant classical-classical OWFs from
quantum-classical OWFs]{Creating quantum-resistant classical-classical OWFs from
quantum-classical OWFs}


\author*[1]{\fnm{Wei Zheng} \sur{Teo}}\email{wt2330@columbia.edu}

\author[2]{\fnm{Marco} \sur{Carmosino}}\email{mlc@ibm.com}

\author[1,2]{\fnm{Lior} \sur{Horesh}}\email{l.horesh@columbia.edu}

\affil*[1]{\orgdiv{Department of Computer Science}, \orgname{Columbia University}, \orgaddress{\street{500 W 120th St}, \city{New York}, \postcode{10027}, \state{NY}, \country{USA}}}

\affil[2]{\orgdiv{Mathematics of AI}, \orgname{IBM Research}, \orgaddress{\street{1101 Kitchawan Rd}, \city{Yorktown Heights}, \postcode{10598}, \state{NY}, \country{USA}}}


\abstract{One-way functions (OWF) are one of the most essential cryptographic primitives, the existence of which results in wide-ranging ramifications such as private-key encryption and proving $\P \neq \NP$ \citep{Impagliazzo89, Goldreich00}. These OWFs are often thought of as having classical input and output (i.e. binary strings), however, recent work proposes OWF constructions where the input and/or the output can be quantum \citep{Buhrman01, Gottesman01, Behera18, Shang20}. In this paper, we demonstrate that quantum-classical (i.e. quantum input, classical output) OWFs can be used to produce classical-classical (i.e. classical input, classical output) OWFs that retain the one-wayness property against any quantum polynomial adversary (i.e. quantum-resistant). We demonstrate this in two ways. Firstly, we propose a definition of quantum-classical OWFs and show that the existence of such a quantum-classical OWF would imply the existence of a classical-classical OWF. Secondly, we take a proposed quantum-classical OWF and demonstrate how to turn it into a classical-classical OWF. In summary, this paper showcases another possible route into proving the existence of classical-classical OWFs (assuming intermediate quantum computations are allowed) using a ``domain-shifting" technique between classical and quantum information, with the added bonus that such OWFs are also going to be quantum-resistant. }


\maketitle

\section{Introduction}
One-way functions (OWFs) are one of the most important entities in computational complexity and cryptography. Proving the existence of classical-classical OWFs (i.e. classical input, classical output) would imply $\P \neq \NP$ as well as prove the existence of private-key cryptography. We first define the notion of a classical-classical OWF in the usual manner \citep{Goldreich00}: 

\begin{definition}
\label{def:owf}
Let $f : \{0,1\}^n \rightarrow \{0,1\}^{p(n)}$ be a deterministic and efficiently computable function, where $p(\cdot)$ is a polynomial. Then $f$ is a one-way function if
\begin{align*}
    \Pr_{x \sim \mathcal{U}_n}[f(A_p(1^n, f(x))) = f(x)] \leq \eps(n)
\end{align*}
where $\mathcal{U}_n$ is the uniform distribution over $n$-bit strings, $\eps(\cdot)$ is a negligible function, and $A_p$ is any probabilistic polynomial time (ppt) algorithm. 
\end{definition}
In other words, given a uniformly random $x$, we obtain $f(x)$ and give it to $A_p$ as input. $A_p$ should not be able to produce an output that is a valid preimage of $f(x)$ with non-negligible probability. Next, we define a quantum-resistant OWF, which is basically the same as Definition \ref{def:owf} but with a quantum polynomial adversary. 

\begin{definition}  
\label{def:qrowf}
Let $f : \{0,1\}^n \rightarrow \{0,1\}^{p(n)}$ be a deterministic and efficiently computable function, where $p(\cdot)$ is a polynomial. Then $f$ is a quantum-resistant one-way function if
\begin{align*}
    \Pr_{x \sim \mathcal{U}_n}[f(A_q(1^n, f(x))) = f(x)] \leq \eps(n)
\end{align*}
where $\mathcal{U}_n$ is the uniform distribution over $n$-bit strings, $\eps(\cdot)$ is a negligible function, and $A_q$ is any quantum polynomial time algorithm. 
\end{definition}

In the following sections, we explore two different ways to construct quantum-resistant OWFs using quantum-classical OWFs. Section 2 proposes a definition of quantum-classical OWFs which can then be easily extended to classical-classical OWFs. Section 3 makes use of the candidate quantum-classical OWF as proposed by \cite{Behera18} to construct an algorithm implementing a classical-classical OWF. The key idea behind these constructions is to base the security of the classical-classical OWFs on the quantum-classical OWFs used in their construction. Since the quantum-classical OWFs are, by definition, secure against quantum polynomial adversaries, then the resulting classical-classical OWFs must be secure against such adversaries as well.  

\section{Classical-classical OWF through composition}
Interactions between classical and quantum computations have the potential to introduce one-wayness. This is demonstrated by the existing constructions for both classical-quantum OWFs \citep{Gottesman01, Buhrman01} (i.e. classical input, quantum output) and quantum-classical OWFs \citep{Behera18} (i.e. quantum input, classical output). \cite{Shang20} proposed a quantum-quantum OWF by simply using a composition of a quantum-classical OWF and a classical-quantum OWF. The input quantum state to the quantum-quantum OWF is first given to the quantum-classical OWF as input to obtain a classical output, and this classical output is then given to the classical-quantum OWF to produce another quantum state as output. Inspired by this approach, we wanted to study whether the same approach can be used to produce a classical-classical OWF -- given a classical input, it is first given to a classical-quantum OWF as input to obtain a quantum state, and this state is then given to a quantum-classical OWF to produce the final classical output. We also investigate the question of whether we need both constituent functions to be one-way or if we can just have one of them be one-way. 

\subsection{Constructing a classical-classical OWF}
Let us consider the following construction of a classical-classical \emph{function}. Let $\mathcal{S} \subseteq \mathbb{C}^{2^m}$ be some finite set of $m$-qubit states. Let $f_1: \{0,1\}^n \rightarrow \mathcal{S}$ be a classical-quantum function. Let $f_2: \mathcal{S} \rightarrow \{0,1\}^{n'}$ be a quantum-classical function. Let $f: \{0,1\}^n \rightarrow \{0,1\}^{n'}$ be defined by $f(x) = f_2(f_1(x))$. For now, we consider only the case where $f_1$ and $f_2$ (and hence $f$) are \textbf{deterministic}. \\

Next, we want to consider what happens if $f_1$ and/or $f_2$ is a OWF. First, we require definitions of classical-quantum and quantum-classical OWFs. 

\begin{definition}
$f_1$ is a classical-quantum OWF if
\begin{align*}
    \Pr_{x \sim \mathcal{U}_n} [f_1(A_q(1^n, f_1(x))) =_k f_1(x)] \leq \eps(n)
\end{align*}
where again $\mathcal{U}_n$ is the uniform distribution over $n$-bit strings, $\eps(\cdot)$ is a negligible function, $A_q$ is any quantum polynomial time adversary, and $=_k$ refers to the outcome where $k$ independent repetitions of the swap test between the two states all pass and $k$ is polynomial sized. 
\end{definition}

It is important to note the use of the swap test in the definition. 
We would like to know whether the $m$-qubit states $f_1(A(f_1(x)))$ and $f_1(x)$ are equal. 
\cite{Gottesman01} proposed to use the \emph{swap test} as the test for equality, as there is no perfect equality test for two quantum states. 
In the swap test, we use a $ \ket{+} $ qubit as the control qubit in a controlled-swap operation between the two states that are to be compared (denoted $ \ket{\psi_1} $ and $ \ket{\psi_2} $). 
We then apply a Hadamard gate to the control qubit and measure it. 
The swap test is passed if we measure $\ket{0}$ and failed otherwise. 
If $\ket{\psi_1} = \ket{\psi_2}$, we will measure $\ket{0}$ with probability 1. 
However, if $\ket{\psi_1} \neq \ket{\psi_2}$ and $\lvert \braket{\psi_1}{\psi_2} \rvert \leq \delta$, then the swap test passes with probability at most $(1+\delta^2)/2$, i.e. the swap test, serving as an equality test, has a probability of failure if the states are unequal. 
With this in mind, it would be ideal for the outputs of $f_1$ to be close to orthogonal, so that the swap test has a larger probability of failing if the states are unequal. 
Using the swap test as the equality test also means that we require a large enough number of copies of the challenge state (say, $k=\omega(\log n)$ copies), such that we can perform enough independent swap tests to eventually end up with a negligible probability that unequal states pass all the swap tests. 
\\

Next, we define a quantum-classical OWF. 

\begin{definition}\label{qcowf-def}

$f_2$ is a quantum-classical OWF if 
\begin{align*}
    \Pr_{\ket{\psi}\sim\mathcal{U}(\mathcal{S})}[f_2(A_q(1^m, f_2(\ket{\psi}))) = f_2(\ket{\psi})] \leq \eps(m)
\end{align*}

where $\mathcal{U}(\mathcal{S})$ represents the uniform distribution over the set $\mathcal{S}$, $\eps(\cdot)$ is a negligible function, and $A_q$ is any quantum polynomial time adversary. Furthermore, for any $\ket{\phi} \not \in \mathcal{S}$, the distribution of $f_2(\ket{\phi})$ should not have any output occurring with non-negligible probability.  
\end{definition}

While we do not need to define an equality test like in the classical-quantum case, we do need to consider the possibility of $A_q$ producing a state not in $\mathcal{S}$. $f_2$ was defined to be deterministic only for states in the domain $\mathcal{S}$. Therefore, on input states outside of $\mathcal{S}$, there is no longer a guarantee that the output will be deterministic. This opens up the possibility of $A_q$ producing a state that is not in $\mathcal{S}$, but when given as input to $f_2$, results in a non-negligible probability of outputting a certain string. When this happens, we can consider this as a kind of `cheating', since $A_q$ did not actually find a preimage in the domain $\mathcal{S}$, but just managed to find some other state that makes $f_2$ output the desired string with non-negligible probability. One way to overcome this would be to simply define a successful inversion where $A_q$ must produce a state in $\mathcal{S}$. However, the definition we adopted was to simply mandate the function to not produce any output with non-negligible probability when given any state not in $\mathcal{S}$, demotivating $A_q$ from `cheating'. Note that there is potential for an extensive discussion on the best way to resolve the issue with `cheating'. For example, we can imagine more relaxed versions of the `no cheating' property, e.g. if a state close enough to the preimage but not in $\mathcal{S}$ also produces the output string with non-negligible probability, we might want to consider that as a successful inversion as well and thus allow such a behavior from $f_2$, while defining it to be difficult to find such a state to retain the one-way property. However, regardless of the method of choice we use to prevent the adversary from `cheating', if we were to consider the use of such a function for the purposes of building a classical-classical OWF as in Theorem \ref{thm:qctocc}, the adversary in the proof of Theorem \ref{thm:qctocc} does not cheat, and hence we can consider definitions without the `no cheating' property (i.e. it is only difficult to find a preimage in $\mathcal{S}$ but it may not be difficult to find a preimage outside of $\mathcal{S}$). The question about the feasibility of alternative quantum-classical OWF definitions and their potential cryptographic applications will be left as an open question for now. The question of whether quantum-classical functions exhibiting the `no cheating' property as defined exists is also left unanswered for now. \\

With these definitions, we can now prove our first theorem. 

\begin{theorem}
If a quantum-classical OWF as defined in Definition \ref{qcowf-def} exists, then a quantum-resistant classical-classical OWF exists.
\label{thm:qctocc}
\end{theorem}

\begin{proof} Recall the construction of $f$ shown above, 
\begin{align*}
    f(x) = f_2(f_1(x))
\end{align*}
where the function signatures are 
\begin{align*}
    f_1 &: \{0,1\}^n \rightarrow \mathcal{S} \\
    f_2 &: \mathcal{S} \rightarrow \{0,1\}^{n'}\\
    f &: \{0,1\}^n \rightarrow \{0,1\}^{n'}
\end{align*}
and $\mathcal{S} \subseteq \mathbb{C}^{2^m}$ is some finite set of $m$-qubit states. 
Let $f_2$ be a quantum-classical OWF. Then $f$ must be a quantum-resistant classical-classical OWF. Suppose $f$ is not a quantum-resistant OWF. Then there exists an adversary $A$ that, when given some $y = f(x)$, is able to find $x'$ with non-negligible probability in quantum polynomial time such that $f(x') = y$. We can then make use of $A$ to construct an algorithm $A'$ that inverts $f_2$ with non-negligible probability. The algorithm for $A'$ is simply to first execute $A$ and obtain a string $x'$ from $A$. With non-negligible probability, we would have $f(x') = y$. Next, we calculate $\ket{\psi} = f_1(x')$. Then $\ket{\psi}$ must be a valid preimage of $y$ with respect to $f_2$, since $f(x') = f_2(f_1(x')) = f_2(\ket{\psi}) = y$. This means that $A'$ has successfully inverted $f_2$ with probability equal to the success probability of $A$ (i.e. non-negligible). Furthermore, $A'$ did not `cheat' and has found a valid preimage of $f_2$ in the domain $\mathcal{S}$, since $\mathcal{S}$ is also the range of $f_1$. We therefore arrive at a contradiction if we assumed that $f_2$ is supposed to be a quantum-classical OWF. Note that we require the assumption that $n'$ should be at least $\omega(\log n)$. If $n' = O(\log n)$, then random guessing can give us a non-negligible probability of inverting $f$, and thus $f$ is not a OWF anymore. 
\end{proof}

Applying a similar construction as the proof above, we do not arrive at a contradiction if we instead set $f_1$ to be a OWF and not $f_2$. Suppose $A$, when given $y = f(x)$, can find an inverse $x'$ such that $f(x') = y$. Let $\ket{\psi} = f_1(x)$, $\ket{\psi'} = f_1(x')$. Clearly $\ket{\psi}$ and $\ket{\psi'}$ can be different, i.e. $x'$ is not necessarily a preimage of $\ket{\psi}$ with respect to $f_1$. Therefore, even if we assume $f$ is not one-way, we do not contradict $f_1$ being a OWF. Note that we have only ruled out one way of constructing classical-classical OWFs from classical-quantum OWFs. Whether classical-quantum OWFs can lead to classical-classical OWFs remains an open question.\\

Next, we aim to explore whether a quantum-classical OWF following the definition in Definition \ref{qcowf-def} can exist. While we do not show with certainty that such a function exists, we explore a non-exhaustive list of ways that do not result in a quantum-classical OWF. 

\subsection{Constructing a quantum-classical OWF}
Indeed, proving the existence of a quantum-classical OWF as defined in Definition \ref{qcowf-def} would almost certainly give us a quantum-resistant classical-classical OWF due to Theorem \ref{thm:qctocc}, provided a suitable and efficient classical-quantum function can be found. For now, we attempt to provide a non-exhaustive list of methods that do not work in giving us a quantum-classical OWF. The methods listed here are generally simple methods which only apply a unitary operation and make some kind of measurement at the end of the unitary operation. 

\subsubsection{How not to construct a quantum-classical OWF}
\begin{enumerate}
    \item \textit{Passing a state through a quantum circuit and deterministically measuring all qubits to produce the classical output}. Using the classical output, we can reconstruct the final state before measurement, and pass this final state through the reverse circuit to obtain the original input state. Therefore such quantum-classical functions are not one-way. 
    \item \textit{Same as 1, but only measuring some of the qubits, and it is easy to fill in for the unknown qubits at the final state}. Suppose we want to use the same method as in 1 to send the final state through the reverse circuit to obtain the original input state. If we only measured (and thus know) some of the qubits at the final state, we will need to find substitutes for the unknown qubits. Now, if it is not difficult to find these substitute qubits, such that the computed input state is in $\mathcal{S}$, then the function is invertible. Note that we observe an interesting circular argument in this construction --- given a circuit that ends off by measuring a portion of the qubits, if this circuit is to implement a quantum-classical OWF, then it must be difficult to find suitable substitute qubits for the unknown qubits at the final state. Observe that the problem of finding the substitute qubits is very similar to the quantum-classical OWF problem statement itself --- given classical information, it is difficult to find a certain suitable state. 
    \item \textit{Have $\mathcal{S} = \mathbb{C}^{2^m}$, and the classical output is obtained by measurement on some or all of the qubits at the end of all quantum computation}. Using the same idea as above (i.e. reconstructing the final state before measurement and sending it through the reverse circuit), we note that if $\mathcal{S} = \mathbb{C}^{2^m}$, then even if there are unknown qubits in the final state, using any kind of qubits for the substitute qubits will still allow the computed input state to be in $\mathcal{S}$. Note that this still holds true even if the circuit implements a probabilistic function, although so far we have mainly assumed that quantum-classical OWFs are deterministic. 
\end{enumerate}

Note that methods 1 and 2 above also do not satisfy the `no cheating' property as described in Definition \ref{qcowf-def}. This is because for a state $\ket{\psi'}$ not in $\mathcal{S}$ but close to some other state $\ket{\psi}$ in $\mathcal{S}$, we can still have a high (non-negligible) probability of producing the same measurements that $\ket{\psi}$ produces, simply because we are only executing a simple unitary operation $U$, and measuring $U\ket{\psi'}$ has a high probability of projecting to $U\ket{\psi}$ if $\ket{\psi}$ and $\ket{\psi'}$ are close.

\section{Classical-classical OWF from candidate quantum-classical OWF} 
In this section, we will mainly be building on the candidate quantum-classical OWF as proposed by \cite{Behera18} to create a classical-classical quantum-resistant OWF. We first briefly go through the candidate quantum-classical OWF, before describing how to turn it into a classical-classical OWF. 

\subsection{Candidate quantum-classical OWF}

\cite{Behera18} proposed a quantum-classical OWF that is different from Definition \ref{qcowf-def}. In this approach, the OWF itself is defined dynamically based on the input state. To formally define this OWF, we have to first introduce GCH states. An $n$-qubit GCH state is a tensor product (in any order) of qubits in either the standard basis $\mathcal{C} = \{\ket{0}, \ket{1}\}$, the Hadamard basis $\mathcal{H} = \{\ket{+}, \ket{-}\}$, or GHZ states containing 2 to $n$ qubits $\mathcal{G} = \bigcup_{j=2}^n \mathcal{B}_j$, where $\mathcal{B}_j = \{\frac{1}{\sqrt{2}}(\ket{x} + \ket{\bar{x}}) \mid x \in \{0,1\}^j\}$. Now we can define the function signature of the candidate quantum-classical OWF: \[\mathcal{F}: \{\ket{\psi^{(n)}}_{GCH}\} \rightarrow \{0,1\}^L,\] where $\{\ket{\psi^{(n)}}_{GCH}\}$ refers to the set of $n$-qubit GCH states. 

\begin{theorem}
Given an input state $\ket{\psi} \in \{\ket{\psi^{(n)}}_{GCH}\}$, we can define a function $\mathcal{F}$ based on $\ket{\psi}$ that is one-way, i.e. for any efficient adversary $A$ and public parameters $P$ (which includes $1^n$, the classical output $\mathcal{F}\ket{\psi}$ and the implementation of $\mathcal{F}$),  
\begin{align*}
    \Pr(\mathcal{F}\ket{\phi} = \mathcal{F}\ket{\psi} \mid A(P) = \ket{\phi}) \leq \eps(n)
\end{align*}
where $\eps(\cdot)$ is a negligible function. 
\end{theorem}

Here we briefly explain the construction of such an $\mathcal{F}$ for a given input state. 
\begin{itemize}
    \item We generate a circuit by adding several layers of parallel CNOT gates (called OWF gate operations) which are randomly placed while following a set of rules. CNOT gates can only be used on two \emph{compatible} qubits, such that if a GCH state is given to an OWF gate operation as an input, then the output is also a GCH state. 
    \item In the final OWF gate operation, we mark one qubit per CNOT gate using some fixed rules based on the type of qubit at each position, and choose $(\frac{n}{2}-1)$ of the marked qubits for measurement. The rules ensure that the marked qubits will either be $\mathcal{C}$ or $\mathcal{H}$ qubits.  
    \item Finally, we perform a POVM measurement on these marked qubits (hence not collapsing the state with measurement) to obtain the classical output. The classical output describes the state of each of the marked qubits (i.e. one of $\ket{0}, \ket{1}, \ket{+}, \ket{-}$). 
    \item We repeat the whole circuit generation procedure for $n$ times (which is possible since the measurement did not collapse the state, hence we can apply the reverse circuit to obtain the original input), i.e. we have $n$ such randomly generated circuits to produce $n$ sets of classical outputs. The final output of $\mathcal{F} \ket{\psi}$ is the concatenation of all $n$ sets of classical outputs. 
    \item The user is given the classical output as well as all $n$ circuits used. It will be difficult for the user to find an input GCH state by various methods of random guessing that produces the same output, except with negligible probability. 
\end{itemize} 

In short, the function involves $n$ repetitions of passing the input state through a circuit, measuring, and passing the state through the reverse circuit, where $n$ is the number of input qubits. In fact, we can imagine this whole process as a \emph{single circuit with intermediate non-destructive measurements}, and therefore future mentions of the term `circuit' will be referring to the one implementing the entire $\mathcal{F}$ instead of just one of the circuits used in some iteration of the algorithm. 

\subsection{Constructing a classical-classical OWF from the proposed quantum-classical OWF}

The first thing to consider is to find a way to encode the input GCH states in binary. Fortunately, this can be easily done. One way to do this is to encode the circuit producing a GCH state in binary. Figure \ref{fig:sampleencoding} shows a possible circuit for encoding a GCH state. We can encode such a circuit in the following manner: 
\begin{itemize}
    \item First $n$ bits: describes the initial qubits ($\ket{0}$ or $\ket{1}$). 
    \item Next $n$ bits: describes which positions have the Hadamard gate applied. 
    \item Next $O(n \log n)$ bits: describes all the positions of CNOT gates. We can enforce certain rules for valid classical encodings at this part. Firstly, the control qubit of each CNOT gate must be on a qubit that is currently a $\ket{+}$. Secondly, for a set of qubits that are supposed to be entangled in the same GHZ state, all the CNOT gates used for entangling all of these qubits must have their control qubit be the qubit with the smallest-numbered position. For example, if qubits 6, 7, 8 are to be entangled, we must have two CNOT gates with (control, target) qubits be (6, 7) and (6, 8). Enforcing these rules will give us a bijective mapping from the set of valid encodings to the GCH states. Lastly, we use $O(n \log n)$ bits because there are at most $n$ CNOT gates since each qubit can only be the target qubit of each CNOT gate once, and the location of each CNOT gate can be described by two $\lceil \log n \rceil$ strings for the position of the two qubits. More precisely, we use at most $2n \lceil \log n \rceil = O(n \log n)$ bits for this section. 
\end{itemize}

\begin{figure}[h]
    \centering
    \includegraphics[scale=0.4]{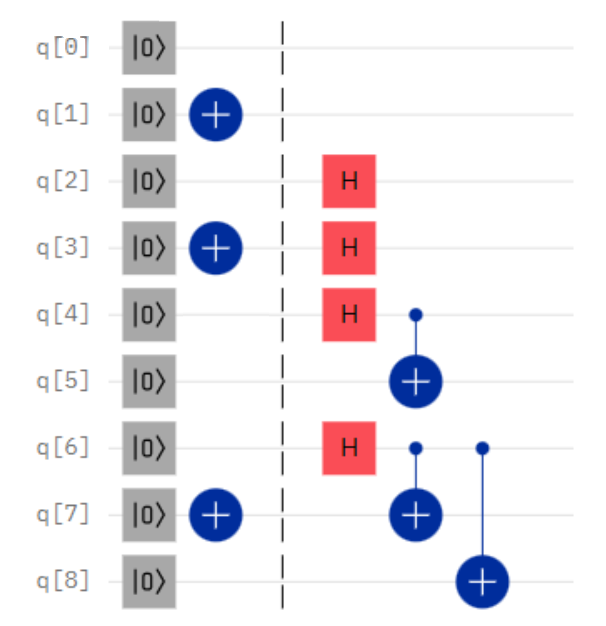}
    \caption{Circuit producing the GCH state $\ket{0}\ket{1}\ket{+}\ket{-}(\ket{00}+\ket{11})(\ket{010}+\ket{101})$ (constant factors ignored)}
    \label{fig:sampleencoding}
\end{figure}

The next obstacle to overcome when converting the quantum-classical OWF of \cite{Behera18} to a deterministic classical-classical OWF is that, given a  function $\mathcal{F}$ sampled according to the proposed algorithm, this function produces \textbf{deterministic output on states with the same basis as the chosen input state} but may produce \textbf{non-deterministic output on states with a basis different from the chosen input state}. At this stage, we make the following key observations: 
\begin{itemize}
    \item A circuit (where `circuit' refers to the entire implementation of $\mathcal{F}$ as explained in the previous subsection) sampled based on an input state $\ket{\psi}$ has deterministic output on all states with the same basis as $\ket{\psi}$. Therefore, if we are able to use a \emph{circuit family} instead of just a single circuit, where each circuit in the circuit family caters to states of a certain basis, then the function implemented by the whole circuit family would be entirely deterministic. 
    \item Sampling the circuit is efficient since its size complexity is $O(n^3)$ as proved by \cite{Behera18}, and so we only have to perform $O(n^3)$ samplings. This means that at runtime, it is efficient to generate the appropriate circuit from the circuit family based on the basis of the input state. Clearly it is also efficient to know the basis of the input state since we need to generate the actual quantum state, and knowledge of the state automatically gives us knowledge of the basis. 
\end{itemize}
Using these observations, we come up with the following algorithm. We require the use of a seeded pseudorandom number generator (PRNG) that we denote $\mathcal{P}_s$, where $s$ is the seed. 

\hypertarget{algo}{}
\begin{algorithm}[H]
\textbf{Input}: A valid encoding $x$ of a GCH state \\
$\ket{\psi} \leftarrow $ build circuit from $x$ and generate GCH state; \\
$b \leftarrow $ basis of $\ket{\psi}$; \\
$\mathcal{F}_b \leftarrow $ execute algorithm from \cite{Behera18} to sample circuit using $\mathcal{P}_b$ as source of randomness \hspace{25pt} \textcolor{darkgreen}{// i.e. all states with the same basis $b$ produce the same circuit} \\
$y \leftarrow \mathcal{F}_b \ket{\psi}$ \hspace{80pt} \textcolor{darkgreen}{// classical output from quantum-classical OWF} \\
$y' \leftarrow y + $ encoding of $\mathcal{F}_b$ \\
\textbf{Output}: $y'$
\caption{Classical-classical OWF from candidate quantum-classical OWF}
\end{algorithm}

The only part of Algorithm \hyperlink{algo}{1} that has not been explained yet is the addition of an encoding of $\mathcal{F}_b$ to the output $y$. The adversary considered by \cite{Behera18} does indeed not only have access to the classical output, but also the implementation of the circuit $\mathcal{F}$. Therefore, by including information about the circuit in $y'$, we ensure that our adversary has access to the same information as the one considered by \cite{Behera18}, allowing us to simply adopt their security guarantee. Another benefit of appending the encoding of $\mathcal{F}_b$ to the output is to prevent collisions. Since states with different bases will produce different circuits, the encoding of these circuits will also be different, which means the only possible collisions must come from states of the same basis. However, these collisions should occur on a low enough frequency such that the one-wayness of the quantum-classical function is not affected, otherwise this would contradict the results of \cite{Behera18}. Lastly, we do not cover in detail how to encode the circuit $\mathcal{F}_b$, since it should be clear that such a polynomial-sized circuit can be encoded efficiently. One possible method of encoding a constituent circuit (instead of the whole circuit representing $\mathcal{F}$; these encodings of each consituent circuit can then be concatenated eventually) is to simply produce a list of CNOT gate positions ordered first by layer, then by lexicographic order according to the control and target positions of the CNOT gates in each layer. The marked and measured qubits at the end of the circuit can each be encoded using $n$ bits. Using this method of encoding will give us a unique encoding for each circuit. 
In short, Algorithm \hyperlink{algo}{1} creates the input state, generates the appropriate circuit, executes the circuit to obtain the classical outputs, then returns the classical outputs along with an encoding of the circuit. If Algorithm \hyperlink{algo}{1} was not one-way, we can clearly see that we will be able to violate the one-wayness of the candidate quantum-classical OWF and contradict the results of \cite{Behera18}. This leads us to our final theorem: 

\begin{theorem}
    If the quantum-classical OWF by \cite{Behera18} is secure against a quantum polynomial adversary, then the resulting classical-classical function from Algorithm \hyperlink{algo}{1} is also a quantum-resistant OWF. 
\end{theorem} 

\begin{proof}[Proof sketch]
Suppose the candidate quantum-classical OWF by \cite{Behera18} is secure against a quantum polynomial adversary, and the classical-classical function from Algorithm \hyperlink{algo}{1} is not a quantum-resistant OWF. Then a quantum polynomial adversary $A$ can do the following steps to invert the candidate quantum-classical OWF: 

\begin{itemize}
    \item {\small $A$ knows the classical output and the circuit used. $A$ encodes the circuit in the same way that Algorithm \hyperlink{algo}{1} encodes the circuit.} 
    \item {\small $A$ now has a valid classical output for the classical-classical function implemented by Algorithm \hyperlink{algo}{1}. Since this function is assumed to be not a quantum-resistant OWF, $A$ is able to obtain a valid preimage with non-negligible probability. }
    \item {\small $A$ can then use the valid preimage to produce a GCH state that is a valid preimage for the candidate quantum-classical OWF, i.e. inverting it. This contradicts the assumption that the candidate quantum-classical OWF is secure. }
\end{itemize}
\end{proof}

\section{Discussion}
We have seen above that quantum-classical OWFs can give rise to quantum-resistant classical-classical OWFs, which can ultimately lead to important breakthroughs in proving the existence of private key cryptography and other computational complexity problems such as proving $\P \neq \NP$. \\

In Section 2, we attempted to produce a classical-classical OWF by composing a classical-quantum and a quantum-classical function, where the quantum-classical function is one-way. We also introduced the idea of a `no cheating' property for such quantum-classical OWFs. While a `no cheating' property is not required if this quantum-classical OWF is only going to be used in the composition of a classical-classical OWF, such a property could still have important cryptographic properties that could be useful in other contexts. We also currently do not have any potential candidates for such quantum-classical OWFs, so this would be a possible area of further research. \\

In Section 3, we showed that if the candidate quantum-classical OWF proposed by \cite{Behera18} is secure against a quantum polynomial adversary, then a quantum-resistant classical-classical OWF exists. The main gap with this result is showing that the candidate quantum-classical OWF is indeed secure. The proof presented by \cite{Behera18} only considers a few ways an adversary might perform an inversion, which mainly involves random guessing and a negligible success probability. However, this is far from a general proof of security, which requires proving the function to be secure against any efficient adversarial algorithm. One way to do this is to show that being able to invert the candidate quantum-classical OWF would violate a proven cryptographic theorem. Even if we only show that inverting the candidate quantum-classical OWF violates certain hardness assumptions (e.g. by reducing some hard problem like LWE to inverting the candidate quantum-classical OWF), it will be useful to better understand the difficulty of inversion. Therefore, we take this opportunity to highlight the need to further study the one-wayness property of functions that are produced in a manner similar to that of \cite{Behera18}, and determine whether such functions can indeed by secure against an arbitrary quantum polynomial adversary. 

\section{Future Work}
In this paper, we have shown that quantum-classical OWFs of varying definitions have the potential to give us quantum-resistant OWFs. We have also surfaced the following problems that are worthy of further research: 
\begin{itemize}
    \item Do classical-quantum OWFs imply classical-classical OWFs? 
    \item What other definitions of quantum-classical OWFs have potentially useful cryptographic applications?
    \item Does there exist a function satisfying Definition \ref{qcowf-def} or any other similar definitions (e.g. with varying definitions of `no cheating')? 
    \item Is the quantum-classical OWF proposed by \cite{Behera18} truly one-way against any quantum polynomial adversary? 
\end{itemize}
Resolving either of the last two questions would most likely give us a quantum-classical OWF and subsequently a quantum-resistant classical-classical OWF. 




\end{document}